\documentclass[11pt]{article}

\usepackage[left=1.25in,top=1in,right=1.25in,bottom=1in]{geometry}
\usepackage{graphicx}
\usepackage{setspace}
\usepackage[round]{natbib}
\usepackage{amsmath}
\usepackage{amssymb}
\usepackage{amsbsy}
\usepackage{amsthm}
\usepackage{paralist}
\usepackage{bm}
\usepackage{booktabs}
\usepackage{setspace}

\usepackage{mathtools}
\usepackage{cases}
\usepackage{algorithm}
\usepackage{algpseudocode}

\newtheorem{theorem}{Theorem}

\newtheorem{corollary}[theorem]{Corollary}

\newtheorem{lemma}[theorem]{Lemma}

\defaultenum{(1)}{}{}{}
\defaultleftmargin{3em}{}{}{}

\newcommand{\nc}{\newcommand}
\nc{\norm}{\mathcal{N}}
\nc{\gig}{\mathrm{GIG}}
\nc{\ig}{\mathrm{IN}}
\nc{\E}{\mathbb{E}}
\nc{\btheta}{{\bf{x}}}

\nc{\by}{{\bf y}}
\nc{\dd}[2]{\frac{\partial #1}{\partial #2}}
\nc{\bone}{{\bf 1}}
\nc{\bbI}{\mathbb{I}}

\nc{\bx}{\mathbf{x}}

\newcommand{\I}{\mathbb{I}}
\nc{\R}{\mathbb{R}}

\begin{document}

\title{Vertical-likelihood Monte Carlo}

\author{Nicholas G. Polson\footnote{University of Chicago Booth School of Business, ngp@chicagobooth.edu}\\
James G. Scott\footnote{University of Texas at Austin, james.scott@mccombs.utexas.edu.}
}

\date{June 2015}

\maketitle

\begin{abstract}

In this review, we address the use of Monte Carlo methods for approximating definite integrals of the form $Z = \int L(x) d P(x)$, where $L$ is a target function (often a likelihood) and $P$ a finite measure.  We present \textit{vertical-likelihood Monte Carlo}, which is an approach for designing the importance function $g(x)$ used in importance sampling.  Our approach exploits a duality between two random variables: the random draw $X \sim g$, and the corresponding random likelihood ordinate $Y\equiv L(X)$ of the draw.  It is natural to specify $g(x)$ and ask: what is the the implied distribution of $Y$?  In this paper, we take up the opposite question: what should the distribution of $Y$ be so that the implied importance function $g(x)$ is good for approximating $Z$?  Our answer turns out to unite seven seemingly disparate classes of algorithms under the vertical-likelihood perspective: importance sampling, slice sampling, simulated annealing/tempering, the harmonic-mean estimator, the vertical-density sampler, nested sampling, and energy-level sampling (a suite of related methods from statistical physics).  In particular, we give an alterate presentation of nested sampling, paying special attention to the connection between this method and the vertical-likelihood perspective articulated here.  As an alternative to nested sampling, we describe an MCMC method based on re-weighted slice sampling.  This method's convergence properties are studied, and two examples demonstrate the promise of the overall approach.

\vspace{0.5\baselineskip}
\end{abstract}

\newpage
\begin{spacing}{1.1}

\section{Introduction}

\subsection{Monte Carlo approximation of integrals}

Suppose we must approximate a definite integral of the form $Z = \int_{\mathcal{X}} L(x) \ d P(x)$, where $L(x)$ is some target function and $P(x)$ a finite measure.   In statistics, these integrals arise in the classical treatment of random-effects models, and in the analysis of incomplete or missing data.  They are also important intermediate quantities in Bayesian model selection, where $L(x)$ is the likelihood, $d P(x) = p(x) \ dx$ the prior, and $Z$ the marginal likelihood (or evidence).  Similar problems arise in statistical mechanics, where $\mathcal{X}$ is usually a discrete space, $P(x)$ is the counting measure, and $Z$ is called the partition function.

The classic Monte Carlo schemes for estimating $Z$ is importance sampling, which involves a weighted sum of likelihood evaluations at points $x^{(i)}$ drawn from a proposal $g(x)$:
\begin{equation}
\label{eqn:impsampling}
\hat{Z} = \sum_{i=1}^N q^{(i)} L(x^{(i)}) \; , \quad x^{(i)} \sim g(x) \, .
\end{equation}
In this paper, we advance a simple principle for choosing the proposal distribution $g(x)$ and calculating the weights $q^{(i)}$.

A major focus of our attention is the dual relationship between two random variables: $X \sim g(x)$, the random draw itself; and $L(X)$, the corresponding random likelihood ordinate of this draw.  Any choice of proposal $X \sim g(x)$ implies some distribution for $L(X)$.  We appeal to the reverse direction: namely, that specifying the distribution of the random likelihood ordinate $Y \stackrel{d}{=}  L(X)$ also implies a proposal distribution $g(x)$ that can be used in (\ref{eqn:impsampling}).  This paper systematically addresses the question: what should the distribution of $Y$ be in order to ensure that the corresponding proposal is a good one for estimating $Z$?

Our answer to this question turns out to unite a wide variety of seemingly disparate methods under a single conceptual framework.  We mention, in particular, the following seven:
\begin{compactenum}
\item importance sampling.
\item slice sampling \citep[e.g.][]{damien:etal:1999,neal:2003}.
\item methods based on powering down/annealing the likelihood or posterior, such as simulated annealing, the power-posterior method \citep{friel:pettitt:2008}, bridge sampling \citep{meng:wong:1996}, and path sampling \citep{gelman:meng:1998}.
\item the harmonic-mean estimator of \citet{newton:raftery:1994}.
\item the vertical-density sampler of \citet{troutt:1993}.
\item energy-level sampling, a generic term that we use to refer to a suite of related methods used in statistical physics \citep[e.g.][]{berg:neuhaus:1991,hesselbo:stinchcombe:1995,wang:landau:2001a,kou:zhou:wong:2006}.
\item the nested-sampling algorithm of \citet{skilling:2006}.
\end{compactenum}
Our goal in this paper is to identify and explain these surprising connections, especially the one with nested sampling.  We are therefore highly selective in our coverage of the literature.  Our goal is not to provide a comprehensive review of Monte Carlo integration.  Rather, it is to provide the reader with a single unifying principle for understanding certain major ideas in this area.

The main requirement of the vertical-likelihood approach is the choice of a weight function $w: \mathbb{R}^+ \to \mathbb{R}^+$, which serves the purpose of re-weighting the implied distribution of likelihood ordinate $Y$.  Our approach differs from traditional importance sampling, in that the proposal $g(x)$ is not defined explicitly, but rather implicitly via $w$.  Because this directs focus away from $g(x)$ and on to the distribution of the likelihood ordinate, we refer to it as vertical-likelihood perspective on Monte Carlo.  By this term, we do not mean a specific algorithm, but rather a perspective on designing proposal distributions for approximating definite integrals. 

The rest of the paper is organized as follows.  The remainder of Section 1 sets the notation, provides some necessary background and definitions, and proves a simple lemma that will be useful in subsequent sections.  Section \ref{sec:statmech} provides a soft introduction to the approach by reviewing several related algorithms from statistical physics.  Our goal here is to articulate a previously unappreciated connection between two historically distinct classes of methods: energy-domain algorithms, and temperature-domain algorithms based on latent variables.  This synthesis serves as the starting point of our approach.

Section \ref{sec:generalapproach} describes vertical-likelihood Monte Carlo in general terms, by relating the choice of importance function $g(x)$ to the one-dimensional weight function $w(u)$.  Several existing methods are shown to be special cases of the approach, corresponding to specific choices of this weight function.  Section \ref{sec:weightfunction} argues in favor of a particular principle to guide the choice of weight function; we call this principle the ``score-function heuristic.''  It also describes a weighted slice-sampling approach for implementing the method.  Section \ref{sec:nested} connects our approach with the nested-sampling algorithm of \citet{skilling:2006}.  Section \ref{sec:examples} gives two examples---a toy one-dimensional problem, and a much harder 50-dimensional problem---showing the excellent performance of the method.  Section \ref{sec:mixing} discusses the mixing properties of our MCMC sampler.  Section \ref{sec:discussion} concludes with some final remarks regarding possible extensions of the approach.

\subsection{Notation, background, and preliminaries}

Let $x \in \mathcal{X}$ be a $d$-dimensional variable, let $L: \mathcal{X} \rightarrow \R^+$ be the target function or likelihood, and let $P$ be a finite measure over $\mathcal{X}$.  Assume without loss of generality that $P(\mathcal{X}) = 1$, so that $P$ is a probability measure.

The quantity of interest is the normalizing constant or partition function
\begin{equation}
\label{eqn:normconstant}
Z = \int_{\mathcal{X}} L(x) \ d P(x) \, .
\end{equation}
One traditional approach for approximating $Z$ is importance sampling, which exploits the identity
$$
Z = \int_{\mathcal{X}} L(x) \ \frac{p(x)}{g(x)} \ g(x) \ dx \, .
$$
Thus one simulates $x^{(i)}$ from some proposal distribution (or importance function) $g(x)$; evaluates the likelihood of each point; and uses weights $q^{(i)} \propto p(x^{(i)})/g(x^{(i)})$, usually normalized to sum to 1.

The main difficulty of importance sampling is the lack of a generally accepted principle for choosing $g(x)$, which exerts a large effect on the variance of $\hat{Z}$. Many rules of thumb have been discussed in the literature.  See, for example, \cite{geweke:1989}, \citet[Section 4.9]{bergerbook2ed}, \citet[Example 7.12]{robert:casella:2004}, \citet{scottberger06}, and \citet[Section 6.3]{robert:2007}.   None of these guidelines, however, provide any guarantee of practical efficiency.

Our paper addresses this gap by advancing a simple principle for choosing $g(x)$.  The central idea is to re-weight the implied distribution of the likelihood ordinate: that is, the random variable $Y \stackrel{d}{=} L(X)$ defined by drawing $X \sim P(x)$ and evaluating the likelihood at the resulting draw.  Throughout the paper, we will use $x$ to denote an element of the original state space $\mathcal{X}$, and $y$ to denote a likelihood ordinate.

Suppose that $X \sim P(x)$ is a draw from the prior, and consider the relationship between the random variables $X$ and $Y \equiv L(X)$.  Characterizing this relationship requires defining two key functions.  First, there is the \textit{upper cumulant} of the prior, denoted $Z(y)$.  We define this as follows.  Let $F_Y(y) = \mathbb{P}\{L(X) \leq y\}$ be the cumulative distribution function of $Y$.  Now define
\begin{equation}
\label{eqn:Zu}
Z(y) =1-F_Y(y) =  \int_{L(x) > y} d P(x)
\end{equation}
to be the complementary CDF of $Y$.  Clearly $Z(y)$ has domain $y \in \mathbb{R}^+$ and range $s \in [0,1]$ and is nonincreasing in $y$.

The second key function is the pseudo-inverse of $Z(y)$, denoted $\Lambda(s)$ and defined as
\begin{equation}
\label{eqn:LambdaZ}
\Lambda(s) = \sup\{y: Z(y) > s\} \, ,
\end{equation}
which, like $Z(y)$, is nonincreasing.  Intuitively, $\Lambda(s)$ gives the value $y$ such that $s$ is the fraction of prior draws with likelihood values larger than $y$.

The functions $Z(y)$ and $\Lambda(s)$ will play a key role in our discussion.  Specifically, each yields a useful identity for the normalizing constant $Z$ that collapses the full $d$-dimensional integral into a one-dimensional integral. We collect these two identities in the following lemma.  Below and throughout the paper, we use $\I$ to denote the indicator function: $\I(a) = 1$ if $a$ is true, and $0$ otherwise.  We also implicitly assume the necessary conditions to permit the interchange of differentiation and integration, and to switch the order of integration inside an iterated integral.

\begin{lemma}
\label{lemma:Zidentities}
The normalizing constant $Z$ in (\ref{eqn:normconstant}) can be expressed in two alternate ways:
$$
Z = \int_0^{\infty} Z(y) \ dy = \int_0^1 \Lambda(s) \ ds \, .
$$
\end{lemma}

\begin{proof}
For the first identity, we have
\begin{eqnarray}
Z &=& \int_{\mathcal{X}} L(x) \ d P(x) \nonumber \\
&=& \int_{\mathcal{X}} \int_0^{\infty}  \I\{y < L(x)\}  \ dy \ d P(x) \nonumber \\
&=& \int_0^{\infty} Z(y) \ dy  \label{eqn:Zin1D} \, 
\end{eqnarray}
by interchanging the order of integration.  Although $Z(y) = 0$ for all $y > \sup_{\mathcal{X}} L(X)$, we still write this integral with an infinite upper boundary to maintain a consistent notation across problems.

The second identity can be derived from the first, by exploiting the fact that $s < Z(y)$ if and only if $y < \Lambda(s)$:
\begin{eqnarray}
Z &=&  \int_0^{\infty} Z(y) \ dy \nonumber \\
&=&  \int_0^{\infty} \int_0^1 \mathbb{I}\{ s < Z(y) \} \ ds \ dy \nonumber \\
&=&  \int_0^{\infty} \int_0^1 \mathbb{I}\{ y < \Lambda(s) \} \ ds \ dy \nonumber \\
&=&   \int_0^1 \Lambda(s) \ ds \label{eqn:Zin1D2} \, .
\end{eqnarray}
\end{proof}

The second identity is important in nested sampling \citep{skilling:2006}, a point which we shall revisit at length in Section \ref{sec:nested}.

\section{The connection with statistical mechanics}

\label{sec:statmech}

We begin with the case of a discrete space $\mathcal{X}$, motivated by the origins of the Monte Carlo method in statistical mechanics.  Let $x \in \mathcal{X}$ be a discrete state variable, and let $h(x)$ be an energy function or Hamiltonian of the state.  Throughout, we will use the example of the Boltzmann distribution
$$
p_T(x) = \frac{1}{Z(T)} \exp\{-h(x)/T\} \, ,
$$
which describes the behavior of a thermodynamical system in thermal equilibrium at temperature $T$.

\citet{kou:zhou:wong:2006} distinguish between two types of features associated with the Boltzmann distribution: temperature-domain features, which are functions of $T$; and energy-domain features, which are functions of the energy level $s$.  Historically, separate classes of algorithms have been used to estimate these two kinds of features.  Our proposed approach combines aspects of algorithms from both domains.  Specifically, it involves the use of auxiliary variables (historically associated with temperature-domain methods), together with the idea of rebalancing the sampler toward higher energy levels (historically associated with energy-domain methods).  We therefore give a brief review of these two sets of ideas, before remarking on an important connection between them.

\paragraph{Temperature-domain features.}   Examples of temperature-domain features, which are functions of $T$, include the partition function
$$
Z(T) = \sum_{x \in \mathcal{X}} \exp\{-h(x)/T\} \, ,
$$
and the Boltzmann average
$$
\mu_f(T) = \E \{ f(x) \}  = \frac{1}{Z(T)} \sum_{x \in \mathcal{X}} f(x)  \exp\{-h(x)/T\} 
$$
of some state function $f(x)$.  These features can be used to calculate certain thermodynamic properties of the system, such as the free energy and the specific heat.

During the second half of the 20th century, much work in statistical mechanics focused on the use of Monte Carlo methods for approximating temperature-domain features.  Most notably, both the original Metropolis algorithm and Hastings' modification \citep{metropolis:etal:1953,hastings:1970} were designed to sample from the Boltzmann distribution for fixed $T$ by means of a Markov chain.  Recall the basic algorithm: let $q(y \mid x)$ be a proposal distribution specifying the probability of proposing a move to state $y$, given the current state $x$.  A proposed move from state $x$ to $y$ is accepted with probability
$$
\alpha = \min \left\{ 1, \frac{p_T(y) q(x \mid y)}{p_T(x) q(y \mid x)} \right\} \,
$$
and otherwise rejected, leading to a Markov chain $\{x_1, x_2, \ldots \}$ whose stationary distribution is $p_T$.  Under suitable conditions, the ergodic average $N^{-1} \sum_{i=1}^N f(x_i)$ can be used to approximate the Boltzmann average of $f(x)$.

A well-known problem with the Metropolis--Hastings algorithm arises when $h(x)$ has many local minima, especially minima separated by high-energy (low-probability) barriers.  In such cases the Markov chain can become stuck in a local minimum and fail to generate samples from the correct distribution with practical runtimes.

Many other temperature-domain techniques have been invented to address this problem, such as parallel tempering \citep{geyer:1991}.  Here, we call attention to one especially relevant class of temperature-domain algorithms: those based on the introduction of auxiliary variables. This includes the Swendsen--Wang algorithm \citep{swendsen:wang:1987}, along with other a wide variety of latent-variable MCMC schemes that have been used to overcome slow mixing in lattice models \citep{higdon:1998}.  These methods all involve augmenting the state variable $x$ by an additional set of variables $u$ such that the higher-dimensional joint distribution $p(x,u)$ has the correct marginal distribution in $x$, and then iteratively sampling $(x \mid u)$ and $(u \mid x)$.  This basic scheme underlies slice sampling, along with many useful Markov-chain Monte Carlo samplers in Bayesian inference \citep[e.g.][]{albert:chib:1993,polson:scott:windle:2012a}.

\paragraph{Energy-domain features and re-balancing schemes.}

Examples of energy-domain features include the microcanonical distribution over the equi-energy surface $\{x: h(x) = s\}$, as well as the microcanonical average of a state function,
$$
\nu_f(s) = \E \{f(x) \mid h(x) = s\} \, ,
$$
which is independent of temperature.  Another energy-domain feature is the density of states.  In the discrete case, this is the number of states with a given energy level: $N(s) = \#\{x: h(x) = s \}$.  In the case where $\mathcal{X}$ is a continuous state space, $N(s)$ is the function such that the volume of the set $\{x: h(x) \in (s, s + ds) \}$ is approximately $N(s) ds$.

The statistical-mechanics community has developed a wide class of Monte Carlo methods to approximate energy-domain features.  We refer to these collectively as energy-level samplers.  These methods all share the goal of biasing the draws toward higher-energy states by means of an iterative re-balancing scheme.

To motivate these methods, let $X$ be a draw from the Boltzmann distribution, assuming $T=1$ without loss of generality.   Let $\eta = h(X)$ be the corresponding random energy level, with distribution
\begin{equation}
\label{eqn:energyleveldensity1}
P(\eta = s) = \sum_{x: h(x) = s } e^{-s} = e^{-s} \ N(s) \, .
\end{equation}
The multicanonical sampler of \citet{berg:neuhaus:1991} attempts to rebalance the sampler so that the implied energy distribution becomes flat: $P(\eta = s) \propto \mbox{constant}$.  As (\ref{eqn:energyleveldensity1}) suggests, this is accomplished by sampling states $x$ with weight inversely proportional to the density of states $N(s)$.

If the density of states is unknown, the Wang--Landau algorithm \citep{wang:landau:2001a,wang:landau:2001b} provides a suitable variation.  It involves estimating $N(s)$ via an iterative re-balancing approach, and has been generalized to a wider class of statistical problems \citep{bornn:etal:2012}.  For a short overview of the adaptive Wang--Landau algorithm, see the Appendix.

An even more extreme re-balancing is the $1/k$-ensemble sampler \citep{hesselbo:stinchcombe:1995}.  Let
\begin{equation}
\label{eqn:cumulativedensitystates}
Z(s) = \#\{x: h(x) \leq s\} = \sum_{t \leq s} N(t)
\end{equation}
define the cumulative number of states with energy as least as small as $s$.  In the $1/k$-ensemble sampler, states are sampled with weight proportional to $Z(s)$, rather than $N(s)$ as in the multicanonical sampler.  This makes it even easier for the sampler to traverse high-energy (low-probability) regions of the state space.

\paragraph{The connection with latent-variable methods.}  Here we note a connection between auxiliary-variable methods and energy-level samplers that underlies our recommended approach for choosing an importance function.  Motivated by the latent-variable scheme at the heart of slice sampling \citep{damien:etal:1999}, consider the joint distribution
\begin{equation}
\label{eqn:jointenergy}
p(x,u) \propto w(u) \ \I\{ u \geq h(x) \} \, .
\end{equation}
Let $(X,U)$ be a random draw from this joint distribution, and just as above, consider the implied distribution over the random energy level $h(X)$:
$$
p(h(X) = s) \propto \sum_{x: h(x) = s} w(s) = w(s) N(s) \, .
$$
If $w(s) = e^{-s}$, we recover the canonical ensemble: that is, the distribution over the energy level implied by the original Boltzmann distribution (\ref{eqn:energyleveldensity1}).  On the other hand, if we set $w(s) = 1/N(s)$, we see that $P(h(X) = s)$ is now constant in $s$, as in the multicanonical sampler and Wang--Landau algorithm.  Finally, if we set $w(s) = 1/Z(s)$ as in (\ref{eqn:cumulativedensitystates}), then we obtain
$$
P(h(X) = s) \propto N(s)/Z(s) \, ,
$$
as in the $1/k$-ensemble sampler of \citet{hesselbo:stinchcombe:1995}.

To summarize: many different sampling schemes historically used for energy-domain features can be interpreted as different choices for the weight function in a joint distribution defined via an auxiliary variable (\ref{eqn:jointenergy}).  The particular form of this joint distribution suggests an interesting connection between slice sampling and energy-level sampling that can be usefully exploited.

\section{The vertical-likelihood perspective}

\label{sec:generalapproach}

\subsection{Overview}


The vertical-likelihood approach shares the idea of biasing the sampler towards low-probability regions.  But there are several differences with the energy-level sampling methods just described.  First, the underlying state space $\mathcal{X}$ is often continuous in statistical applications, which introduces complications not present in the discrete case.  Second, our re-weighting scheme is based on the likelihood ordinate $Y = L(X)$, rather than the energy function, or equivalently the log of the likelihood.  This allows us to connect the vertical-likelihood perspective with many other methods for estimating normalizing constants, including nested sampling.
 Finally, and most notably, our method does not involve an iterative scheme to estimate the density of states, as in the multicanonical or $1/k$-ensemble sampler.  Instead, it can be seen as a generalization of slice sampling, where the slice variable is the analogue of the energy level.  
 
Consider the following latent-variable representation of the likelihood $L(x)$, used both in slice sampling and the Swendsen--Wang algorithm \citep{higdon:1998,damien:etal:1999,neal:2003}:
$$
L(x) = \int_0^{\infty} \I \{0 < u < L(x) \} \ d u = \int_0^{L(x)} d u \, ,
$$
where $\I$ is the indicator function.  This allows us to write the posterior distribution $\pi(x) = p(x) L(x) / Z$ as the marginal of the joint distribution
\begin{equation}
\label{eqn:slicejoint}
\pi(x, u) = \frac{\I\{0 < u < L(x)\} p(x)}{Z} \, ,
\end{equation}
where the latent variable $u$ indexes the likelihood ordinate $L(x)$.  As in (\ref{eqn:jointenergy}), the key step in our approach is the introduction of a weight function whose purpose is to rebalance this joint distribution toward lower likelihood ordinates.

Specifically, let $w: \R^+ \rightarrow \R^+$ be a weight function, and let $W(u) = \int_0^u w(s) \ ds$ be the corresponding cumulative weight function.  Define the weighted joint distribution
\begin{equation}
\label{eqn:weightedslice}
\pi_w(x, u) =  \frac{w(u) \  \I\{0 < u < L(x)\} \ p(x)  }{Z_w} \, ,
\end{equation}
where
\begin{eqnarray}
Z_w &=& \int_{\mathcal{X}} \Omega_w(x) \ p(x) \ d x  \nonumber \\ 
\Omega_w(x) &=& \int_{0}^{\infty} w(u) \ \I\{0 < u < L(x)\} \ d u = \int_0^{L(x)} w(u) \ d u = W(L(x)) \label{eqn:generalOmw} \, .
\end{eqnarray}
Neither the marginals nor the conditionals of $\pi_w(x, u)$ correspond to those of the joint posterior (\ref{eqn:slicejoint}), unless $w(u) = 1$.  Instead, they define a working model whose marginal distributions are easily shown to be
\begin{eqnarray}
\pi_w(x) &=& \frac{p(x) W(L(x))}{Z_w} \label{eqn:marginal1} \\
\pi_w(u) &=& \frac{w(u) Z(u)}{Z_w} \label{eqn:marginal2} \, ,
\end{eqnarray}
where $Z(u) = \int_{L(x) > u} p(x) d x$ is the analogue of the cumulative density of states in (\ref{eqn:cumulativedensitystates}), with the prior distribution $p(x)$ in place of the counting measure.

Our approach is to specify a weight function $w(u)$ in (\ref{eqn:weightedslice}), draw values $x^{(i)} \sim \pi_w(x)$ from the corresponding marginal distribution (\ref{eqn:marginal1}), and form an estimate of $Z$ via importance sampling (\ref{eqn:impsampling}), with weights
\begin{equation}
\label{eqn:sliceweights}
q^{(i)} \propto \frac { p(x^{(i)} ) }{ \pi_w(x^{(i)}) }  \propto [W\{L( x^{(i)} ) \} ]^{-1} \, 
\end{equation}
normalized to sum to 1.  The choice of proposal $g(x)$ over the original space is thereby reduced to the one-dimensional weight function $w(u)$.

\subsection{Special cases}

This approach raises many practical questions, especially regarding the choice of weight function.  These will be addressed in subsequent sections.  First, however, we will show that several existing Monte Carlo methods correspond to special choices of the weight function in (\ref{eqn:weightedslice}).  

We begin with slice sampling.  Observe that the conditionals corresponding to the re-weighted joint distribution are
\begin{eqnarray}
\pi_w(x \mid u) &=& \frac{ p(x) \ \I\{ L(x) \geq u \} } {Z(u) } \label{eqn:conditional1} \\
\pi_w(u \mid x) &=& \frac{w(u) \ \I\{ 0\leq u \leq L(x) \}  }{W(L(x))} \, . \label{eqn:conditional2}
\end{eqnarray}
where $Z(u) = \int_{L(x) > u} p(x) dx$.  Now consider the special case of a uniform weight function, $w(u) \equiv 1$.  With this choice, $W(L(x))$ and $Z_w$ revert to the ordinary likelihood and marginal likelihood, respectively, and the ordinary slice-sampling conditionals are recovered \citep[c.f.][]{damien:etal:1999}.  In this case, the importance function is just the original posterior, the weights in (\ref{eqn:sliceweights}) are $q^{(i)} \propto  1/ L( x^{(i)} )$, and the method reduces to the harmonic-mean estimator \citep{newton:raftery:1994}.

This connection provides insight on the crucial role played by the weight function.  In particular, the harmonic mean estimator of $Z$ typically has infinite variance \citep[e.g.][]{raftery:valencia:2007}, and is known to converge to a one-sided stable law with characteristic exponent $1 < \alpha < 2$ \citep{wolpert:schmidler:2012} under quite general conditions.  The problem with the harmonic mean is that very large changes to the prior $p(x)$ produce correspondingly large changes in $Z$, but usually produce only minor changes in the posterior distribution.  Because the harmonic-mean method uses posterior samples to estimate $Z$, it is inappropriately insensitive to large changes in the prior.  In light of this, it is clear that $w(u)$ should be a decreasing function, so as to heighten the sensitivity of the importance function to the prior.

A second special case is the parametrized weight function $w(u) = au^{a-1}$ for $a \in (0,1)$.  Then $W(u) = u^a$, and the implied importance function is the posterior that arises from using a powered-down version of the likelihood:
$$
\pi_w(x) = \frac{p(x) L(x)^a}{Z_w} \, .
$$
This corresponds to the importance function used in the power-posterior method of \citet{friel:pettitt:2008}, and is similar to simulated annealed and annealed importance sampling \citep{neal:2001}.

This also connects our method with bridge sampling \citep{meng:wong:1996} and path sampling \citep{gelman:meng:1998}.  To see this, consider adding a point mass at $u=0$ to the weight function: $w(u) = c \delta_0 + (1-c) w^\star(u) $ where $\delta$ is a Dirac measure.  If $w^\star(u)=1$ this leads to a mixture of the prior and the annealed posterior distribution as the importance function.  On the other hand, if $w^\star(u)=au^{a-1}$, we have a mixture of the prior and the posterior corresponding to a powered-down version of the likelihood.

The power-posterior and related methods are archetypal of most approaches to importance sampling: they construct a proposal distribution by manipulating the posterior density so that it will retain a similar shape, but with higher variance or heavier tails.  Our representation provides a complementary view of these methods: as members of a wider family of parametric weight functions $w(u)$ that bias the distribution of the likelihood ordinate toward lower values.

\section{The weight function}

\label{sec:weightfunction}

\subsection{A score-function heuristic}

Now consider the choice of weight function $w(u)$.  To motivate our recommended approach, return to the first identity for $Z$ in Lemma \ref{lemma:Zidentities}:
$$
Z = \int_{\mathcal{X}} L(x) \ dP(x) = \int_0^{\infty} Z(y) \ dy \, .
$$
This looks suspiciously like a free lunch, in that the original $d$-dimensional integral has been collapsed into the one-dimensional space of likelihood ordinates $y \in \mathbb{R}^+$.  A reasonable question is: having collapsed the integral into a single dimension, can we just use a simple method, such as Monte Carlo integration or adaptive quadratic?

The answer is a very definite no.  The reason is that the upper cumulant
$$
Z(y) = \int_{x: L(x) > y} d P(x)
$$
changes much more rapidly over its domain than does $L(x) p(x)$.  Moreover, the behavior of the integrand at both boundaries (but especially near the maximum value of the likelihood function) contributes significantly to the value of the integral.  A quadrature method is doomed to failure, unless the grid is chosen with extreme care.  Put simply, $Z(y)$ is much too spiky to estimate this integral by standard methods, even though it is ``only'' one-dimensional.  (In fact, as we describe in the next section, nested sampling can be viewed as a stochastic method of choosing an extremely careful grid for quadrature.)  

As a result, we never use this identity for explicitly calculating $Z$.  We do, however, use it to provide intuition regarding the choice of proposal density.  Specifically, the identity suggests a useful guideline: if $X$ is a draw from the proposal, the corresponding likelihood value $L(X)$ should concentrate with high probability in regions where $Z(y)$ changes rapidly, relative to its value.  Otherwise, we are unlikely to generate samples $y$ in regions of likelihood-ordinate space that contribute the most towards the overall value of the integral $ \int_0^{\infty} Z(y) \ dy$.

This motivates a simple score-function heuristic.  Regions where the score function $d/dy \log Z(y) = Z'(y)/Z(y)$ is large correspond precisely to regions of likelihood-ordinate space where $Z(y)$ changes rapidly relative to its value.  This is where our samples should concentrate in order to ``zoom in'' on $Z(y)$'s largest contributions to the overall value of $Z$.  Therefore, the importance function  $X \sim g(x)$ should be chosen so that the implied distribution of the likelihood ordinate $Y \equiv L(X)$ has density
$$
f(y) \propto \frac{Z'(y)}{Z(y)} \, ,
$$
at least in an approximate sense that we will soon make more precise.

\begin{figure}
\begin{center}
\includegraphics[width=5in]{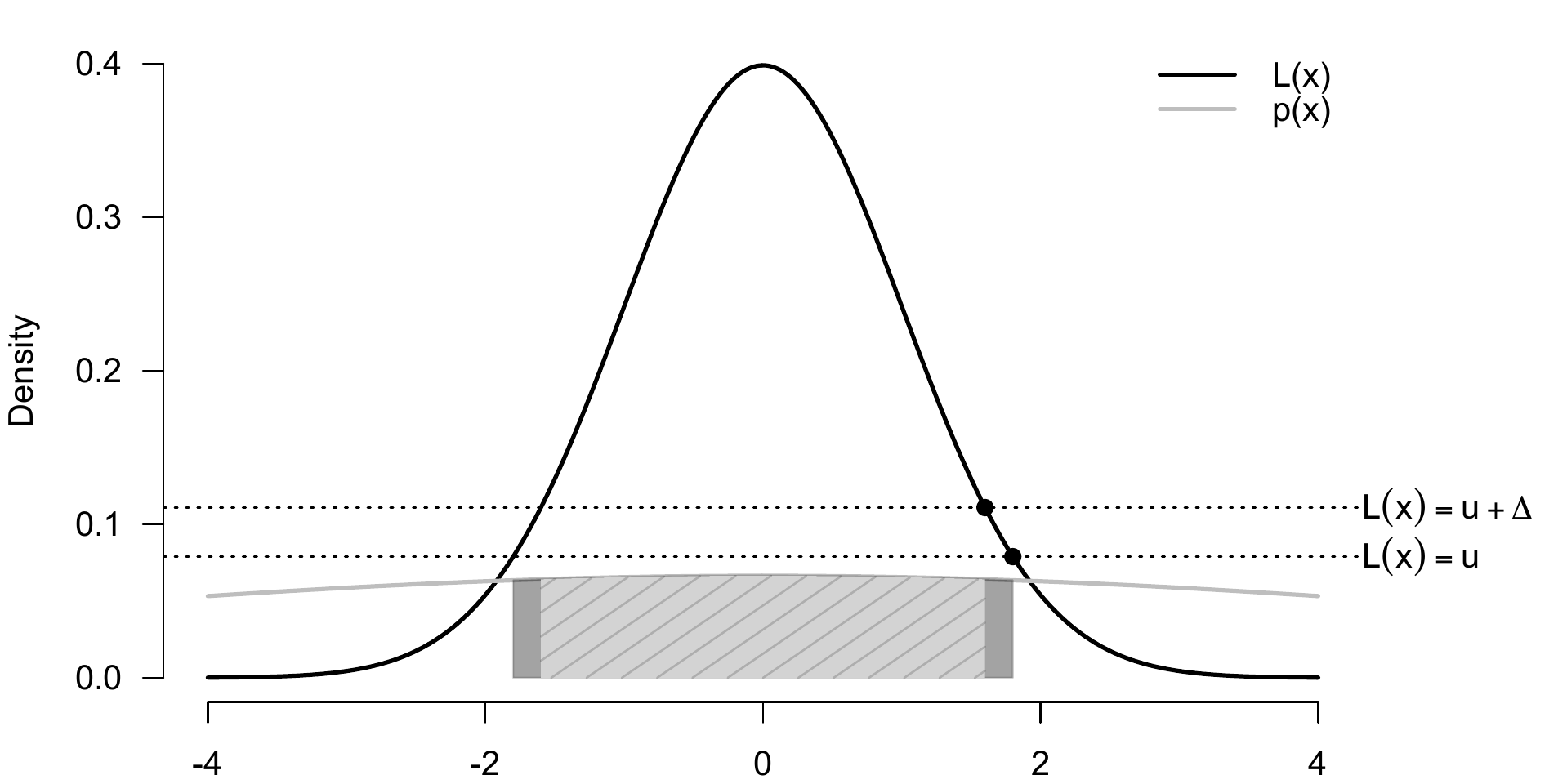}
\caption{
\label{fig:scorefunction}
A graphical depiction of the score-function heuristic.  The light-grey area with diagonal lines shows the quantity $Z(u + \Delta)$, while the dark-grey area shows the difference $Z(u) - Z(u + \Delta)$.
}
\end{center}
\end{figure}

Figure \ref{fig:scorefunction} gives a graphical depiction of the score-function heuristic.  The two horizontal dotted lines depict likelihood slices at $u$ and $u + \Delta$.  The light-grey area with diagonal lines depicts the quantity $Z(u + \Delta)$, while the dark grey area depicts the difference $Z(u) - Z(u+ \Delta)$; thus $Z(u)$ is the sum of the two areas.  The score function heuristic says that we should choose the distribution of the likelihood ordinate $Y$ to satisfy
$$
\mathbb{P}\{Y \in (u, u+\Delta) \} \approx \frac{Z(u) - Z(u + \Delta)}{Z(u)} \approx \Delta \cdot f(u) \, ,
$$
or the ratio of the dark-grey area to the light-grey area with diagonal lines.  In the limit as $\Delta \to 0$, this becomes $f(y) \propto -Z'(y)/Z(y)$.

The following theorem shows how this heuristic can be operationalized.  It characterizes the distribution of the likelihood ordinate under random sampling from the proposal in (\ref{eqn:marginal1}).
\begin{theorem}
\label{theorem:ordinatedistn}
Let $w: \R^+ \rightarrow \R^+$ be a weight function, and let $W(u) = \int_0^u w(s) ds$ be the corresponding cumulative weight function.
Suppose that $X \sim g_w(x)$, defined as in (\ref{eqn:marginal1}) for prior $p(x)$ and likelihood $L(x)$.
Let $Y \equiv L(X)$ be the random likelihood ordinate when $X \sim g_w(x)$, and let $Z(u)$ be the prior measure of the set $\{ x: L(x) > u\}$ as in (\ref{eqn:Zu}), with derivative $Z'(u)$.  Then the density of $Y$ under $g_w$ is
\begin{equation}
\label{eqn:ordinatedistn}
f(y) = - \frac{W(y) Z'(y)} {Z_w} \, .
\end{equation}
\end{theorem}

This result, together with the score-function heuristic, suggests that the cumulative weight function should be inversely proportional to $Z(y)$.  In fact, we suggest the following choice for $\eta \in (0,1]$:
\begin{equation}
\label{eqn:defaultWeightFunction}
W(u) = \frac{1}{\max\{\eta, Z(u)\}}
\end{equation}
This defines a family of importance functions that all meet the score-function heuristic on a restricted range of likelihood ordinates, with the prior ($\eta = 1$) as a boundary case.  Note that the truncation of the denominator at $\eta$ is to ensure the propriety of the corresponding importance function.  As a guideline, we recommend that $\eta$ be chosen so that $Z^{-1}(\eta)$ is very close to the maximum of the likelihood function.  Moreover, as we will soon demonstrate, the connection between nested sampling and our method provides substantial insight regarding the choice of $\eta$.

Theorem 2 provides an interesting generalization of the vertical-density representation of \citet{troutt:1993}.  
Let $f(x)$ be a density, and suppose that $X \sim f$.  Let $S_X(v) = \{x: f(x) > v  \}$, and let $\mu(A)$ denote the Lebesgue measure of the set $A$.  Then the random variable $Y \stackrel{d}{=} f(X)$, the vertical density ordinate of $X$, has density $g(v) = -v Z'(v)$, where $Z(v) = \mu\{S_X(v)\}$.  Theorem \ref{theorem:ordinatedistn} gives an analogous result for likelihood ordinate in the re-weighted joint distribution (\ref{eqn:slicejoint}), with the prior measure of the set $\{x: L(x) > y\}$ replacing Lebesgue measure.

\subsection{One route to implementation: slice sampling}

One possible approach for implementing the proposed method invokes an analogy with slice sampling.  That is, we may iterate between the conditionals $(u \mid x)$ and $(x \mid u)$ that arise from the weighted joint distribution in (\ref{eqn:weightedslice}), and that are given in (\ref{eqn:conditional1})--(\ref{eqn:conditional2}).  This imposes essentially the same requirements of ordinary slice sampling: one must sample from the prior $p(x)$, conditioned to the region $\{x: L(x) \geq u \}$.  See Section \ref{sec:mixing} for a discussion of the mixing properties of the resulting Markov chain.  We also note that as a byproduct, this algorithm can be used to provides an estimate of the entire curve $\hat{Z}(u)$.

The only additional requirement is the ability to sample from the density proportional to the weight function, right-truncated by the current value of the likelihood function.  In the case of a parametric weight function, this will involve sampling a truncated density proportional to $w(u)$.  We do not discuss this step at length, as it will be highly dependent on context.  A quite general approach for sampling from these conditional distributions will involve an application of the Metropolis--Hastings method, as in nested sampling.

For the default weight function described earlier, we have the following lemma that is applicable to cases where either $Z(y)$, or its inverse, can be calculated cheaply.  We omit the proof, which is simple algebra.
\begin{lemma}
\label{lemma:samplingu}
Let $W(u) =1/\max\{\eta, Z(u)\}$, and let $\pi_w(u \mid x)$ be the conditional distribution in (\ref{eqn:conditional2}).  Suppose that $T \sim \mbox{Unif}\left(0,\frac{1}{\max\{ \eta, Z(L(x)) \}} \right)$.  If $T \leq 1$, set $u = 0$.  Otherwise, set $u = Z^{-1} (1/T)$.  The resulting $u$ is a draw from $\pi_w(u \mid x)$.
\end{lemma}
In cases where $Z(y)$ is not easily invertible, but can be evaluated for fixed $u$, the second step can be accomplished fairly cheaply using bisection \citep[Section II.2.2]{devroye:1984}.

To evaluate $Z(y)$, one must compute the prior measure of the set
$$
S(y) = \{x \in \mathcal{X} : L(x) > y \} \, .
$$
For this, we refer the reader to \citet[Section 3.4]{troutt:pang:hou:2004}, who give expressions for the set $S(y)$ for a very large number of common likelihood families, including the multivariate normal, exponential, and logistic distributions on $\mathbb{R}^d$; the multivariate uniform distribution on the sphere; the $L^p$-norm symmetric distributions; and multivariate extensions of the Pareto distribution.  The set $S(y)$, together with the quantile function of the prior $P(x)$, will often be sufficient to implement the weighted slice-sampling scheme just outlined.  See Section \ref{sec:examples} for examples.

\section{The connection with nested sampling}

\label{sec:nested}

\subsection{An alternate presentation of nested sampling}

In this section, we show that nested sampling \citep{skilling:2006} is equivalent to our approach, subject to a choice of the weight function $w(u)$ that nearly replicates our recommended choice.  Indeed, nested sampling can be thought of as one possible algorithm that sequentially generates an importance function satisyfing the score-function heuristic.  Moreover, as we will discuss below, the variation of nested sampling proposed by \citet{brewer:etal:2011} may be the most generally applicable algorithm for implementing our approach with the default weight function.  See also \citet{chopin:robert:2010}, who establish that the algorithm gives estimates of $\hat{Z}$ that are asymptotically Gaussian.

We first give a slightly modified presentation of nested sampling, which is more suited to our purposes than that of the original paper.  Start from the second identity for $Z$ derived in the introduction: $Z = \int_0^1 \Lambda(s) \ ds$.  Now imagine placing down a grid of ordered values $s_1 > s_2 > \cdots > s_n$, and approximating $Z$ using the Riemann sum
$$
\hat{Z} = \sum_{i=1}^n \Lambda(s_i) \{s_{i-1} - s_i\} \, ,
$$
with $s_0 \equiv 1$. The two key things that must be specified are the grid points and the values of the inverse cumulant $\Lambda$ at the grid points.

Nested sampling provides a systematic way of estimating $\Lambda(s)$ at the grid points $s_i = \exp(-i/K)$, where $K$ is modest (say, 20), and $i = 1, \ldots, n$ for large $n$ (say, 1000). This choice leads to the estimate
$$
\hat{Z} =  \sum_{i=1}^n \Lambda(e^{-i/K} ) \left[ e^{-(i-1)/K} - e^{-i/K} \right] \, .
$$
Notice that we accumulate contributions to the integral starting near $z=1$ and move left towards $z=0$, and that the grid points become exponentially closer together as we move toward zero.  (We describe the algorithm via a Riemann sum for the sake of simplicity, but in practice the trapezoid rule should be used instead.)  The key requirement of the algorithm is being able to take draws from the prior distribution $x \sim P$, conditional upon the likelihood of the sampled point exceeding a certain threshold.  The threshold itself is specified recursively, defining a family of distributions
\begin{equation}
\label{nested:nicksrecursion}
P_{(i)}(x) = \frac{P(x) \cdot \mathbb{I}\{ L(x) > y_{i-1} \} }{Z(y_{i-1})} \, .
\end{equation}

\paragraph{First step.}  The first step of nested sampling consists of the following substeps.
\begin{enumerate}
\item Take $K$ draws $x_1, \ldots, x_K$ from the prior $P$.
\item Evaluate their likelihoods, $L(x_1), \ldots, L(x_K)$.
\item Sort the likelihood values and label them $l_{(1)}, \ldots, l_{(K)}$.
\end{enumerate}
Note that substep 1 is the same as drawing from $P_1$ in (\ref{nested:nicksrecursion}), with $y_0 \equiv 0$.

We pay particular attention to the first order statistic (i.e.~the smallest likelihood ordinate), which we label as $y_1$.  A key identity for understanding the relevance of this quantity is the following.  If $X \sim P$ and $U \sim U(0,1)$, it is easy to show that
$$
\Lambda(U) \stackrel{D}{=} L(X) \, .
$$
Equivalently, if $X \sim P$, then
\begin{equation}
\label{eqn:ZLuniform}
Z[L(X)] \stackrel{D}{=} U(0,1) \, .
\end{equation}

Because $x_j \sim P$, Equation (\ref{eqn:ZLuniform}) says that the corresponding values $Z(L(x_j))$ are uniformly distributed.  Therefore, because $Z(y)$ is a decreasing function, $s_1 \equiv Z(y_1) \equiv Z(L(x_{(1)})$---that is, the upper prior cumulant of the smallest likelihood ordinate---is the \textit{maximal} order statistic of a sample of size $K$ from $U(0,1)$.  Therefore $s_1 \sim \mbox{Beta}(K,1)$, and
$$
E(s_1) = \frac{K}{K+1} \approx e^{-1/K} \, .
$$
Accordingly, in expectation,
$$
s_1 = Z(y_1) \approx e^{-1/K} \, ,
$$
or equivalently
$$
\Lambda(e^{-1/K}) \approx y_1 \, ,
$$
expressed in terms of the inverse cumulant function $\Lambda(s)$ defined in (\ref{eqn:LambdaZ}).  Thus $(e^{-1/K}, y_1)$ becomes the first ordinate-abscissa pair used in the Riemann sum to approximate $Z$.

\paragraph{Second step.} In the second step of nested sampling, we draw $K$ samples from $P_2(x)$: that is, the prior distribution, conditional upon $L(x) > y_1$ (\ref{nested:nicksrecursion}). This is efficient, as we only need to discard the point of minimum likelihood we used in step 1, and replace it with a new draw from the prior, conditioned to the region $\{x: L(x) > y_1\}$. The other $K-1$ points from the first step satisfy the constraint by construction.  Thus the required draw can be simulated by, for example, running several MCMC steps starting from one of the other $K-1$ points.

The distribution of $Z(L(x_j))$ is the same uniform distribution as in the first stage, conditioned to the region $\{ Z[L(x_j)] < e^{-1/K} \}$.  This is just a scaled uniform distribution:
$$
Z(L(x_j)) \sim U(0, e^{-1/K}) \, .
$$
Just as before, let $y_2$ be the smallest likelihood value among the sampled points: $y_2 \equiv l_{(1)} = \min\{L(x_1), \ldots, L(x_K)  \}$.  Again via (\ref{eqn:ZLuniform}), $Z(y_2)$ is the maximal order statistic of a sample of size $K$ from $U(0, e^{-1/K})$.  This is just a scaled beta distribution, with expectation $e^{-1/K} \cdot K/(K+1) \approx e^{-2/N}$.  Equivalently, in terms of the inverse cumulant,
$$
\Lambda(e^{-2/K}) \approx y_2 \, .
$$
Therefore the second likelihood ordinate $y_2$ is an estimate of the inverse cumulant $\Lambda(s)$ at the grid point $s_2 = e^{-2/N}$.

\paragraph{Subsequent steps.}  It is easy to check that the argument used to characterize the second step applies recursively to each subsequent step.  Each time we peel away a fraction of $e^{-1/N}$ (in expectation) from the right tail area of the random variable $Z(L(X))$.  Let $y_{i-1}$ be the minimal likelihood value used at step $i-1$.  At step $i$ we generate $N$ samples drawn from the prior distribution, conditioned to the region $\{\theta: L(\theta) > \lambda_{i-1}\}$.   We compute the likelihood values of these sampled points, take the minimal order statistic, and call this $y_i$.  This approximates the inverse cumulant $\Lambda(s)$ at the point $s_i = e^{-i/K}$.  After $n$ repetitions of the whole process, we form the estimated Riemann sum
$$
\hat{Z} = \sum_{i=1}^n y_i \{e^{-(i-1)/K} - e^{-i/K}\} 
$$
using the ordinate-abscissa pairs
\begin{equation}
\label{eqn:nestedLikelihoodDraws}
y_i \approx \Lambda(e^{-i/K})
\end{equation}
approximated at each step.

\subsection{The implied distribution of the likelihood ordinate}

A natural question is: what is the distribution of the likelihood ordinates $y_i$ calculated via nested sampling?  The answer is surprising: it is essentially the same as the distribution of likelihood ordinates implied by recommended weight function from Section \ref{sec:weightfunction}.  This provides us with an alternative interpretation of nested sampling, as a clever means of constructing an importance function whose likelihood-ordinate distribution satisfies the score-function heuristic.

We now sketch out a proof of this fact. Let $\gamma$ be so that $Z(\gamma)$ is exponentially small, such as $\log Z(\gamma) \approx -n/K$, where $n$ and $K \ll n$ are the same as those used in nested sampling.  Intuitively, this says that $\gamma$ is sufficiently high in likelihood-ordinate space so that the vast majority of all priors draws would have lower likelihood.  Define the cumulative weight function
\begin{equation}
\label{eqn:NestedWeightFunction}
W(u) = \frac{1}{Z(u) \log Z(\gamma)} 
\end{equation}
for $u \in [0, \gamma]$, and suppose that this weight function is used to derive an importance function, as in (\ref{eqn:marginal1}).  From Theorem 1, this implies that the corresponding density of the likelihood ordinate $Y \equiv L(X)$ is
\begin{equation}
\label{eqn:defaultYdensity}
f(y) = \frac{Z'(y)}{Z(y) \log Z(\gamma)} \; , \quad y \in [0, \gamma] \, .
\end{equation}

This choice attains the score-function heuristic on the range $y \leq \gamma$.  Moreover, we can verify that this is a properly normalized density function by observing that (\ref{eqn:defaultYdensity}) is the derivative of the function
$$
F(y) = \frac{\log Z(y)}{\log Z(\gamma)} \, .
$$
As $Z(y)$ is a nonincreasing function and $\log Z(\gamma) < 0$, $F(y)$ is nondecreasing.  Moreover, $F(0) = 0$ and $F(\gamma) = 1$.  Therefore $F(y)$ is a valid cumulative distribution function, and $f(y) = F'(y)$ is a valid density.

Let $y_q$ denote the $q$th quantile of $f(y)$ in (\ref{eqn:defaultYdensity}).  That is,
$$
q = P(Y \leq y_q) = \frac{\log Z(y_q) }{\log Z(\gamma)} = \frac{ \log Z(y_q)} {-n/K} \, .
$$
We may express this equivalently in terms of the inverse cumulant as
$$
Z(y_q) = e^{-qn/K} \, .
$$

We can now make the connection with nested sampling explicit.  Now suppose we take $n$ draws $y_1, \ldots, y_n$ directly from (\ref{eqn:defaultYdensity}), and sort the draws in increasing order $y_{(1)}, \ldots, y_{(n)}$.  Clearly the empirical estimate of the quantile $q_i = F(y_{(i)})$ is simply $q_i \approx i/n$.  Thus we have sorted likelihood ordinates $y_{(i)}$ for which
$$
Z(y_{(i)}) = e^{-q_i n/K} \approx e^{-i/K} \; , \quad i = 1, \ldots, n \, , 
$$
or equivalently,
$$
y_{(i)} \approx \Lambda(e^{-i/K}) \, ,
$$
which is identical to (\ref{eqn:nestedLikelihoodDraws}).  In summary, we have shown that using the weight function in (\ref{eqn:NestedWeightFunction}), which very nearly matches the recommended choice from Section \ref{sec:weightfunction}, generates likelihood ordinates and weights that are statistically identical to those generated by nested sampling.  

\citet{skilling:2006} gives examples where nested sampling cannot be expected to work well.  Although we do not pursue the point here, it is possible than an even more extreme choice of weight function may perform better in these situations.  See also \citet{brewer:etal:2011}.

\section{Examples}

\label{sec:examples}

\begin{figure}
\begin{center}
\includegraphics[width=4.0in]{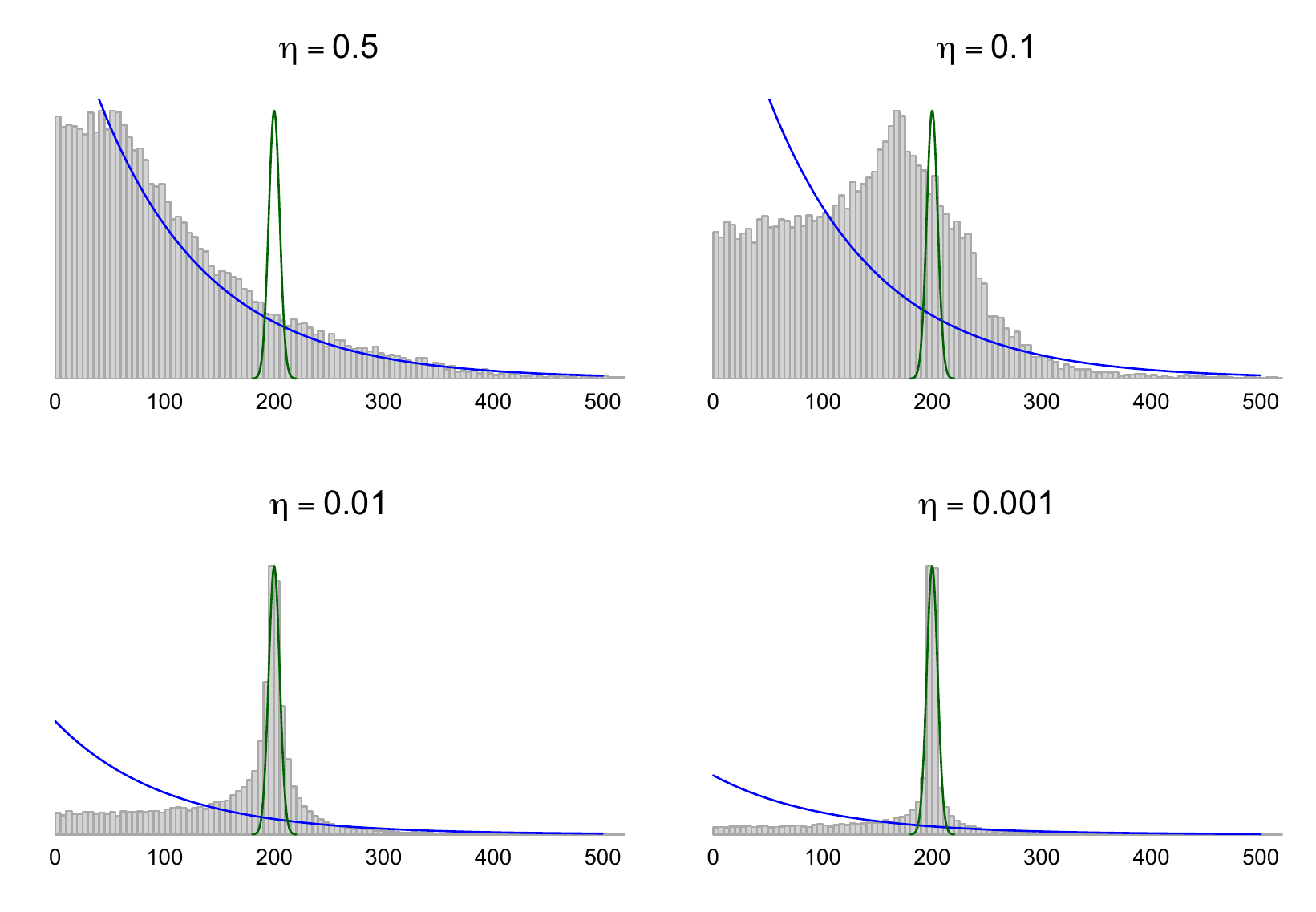}
\caption{\label{fig:normexp2} The implied importance functions for the normal-exponential example ($\tau=100$, $\sigma=5$) under the default weight function, with four different choices of $\eta$.  For $\eta$ near 1,  $\pi_w(x)$ closely resembles the prior.  As $\eta \to 0$, $\pi_w(x)$ becomes more peaked near the region of $L(x)$ is largest, yet always has tails like the prior distribution.  In each panel, the blue line show the density of the prior, while the green line shows the likelihood on a standardized scale.
}
\end{center}
\end{figure}

\subsection{Normal likelihood, exponential prior}
\begin{figure}
\includegraphics[width=6.0in]{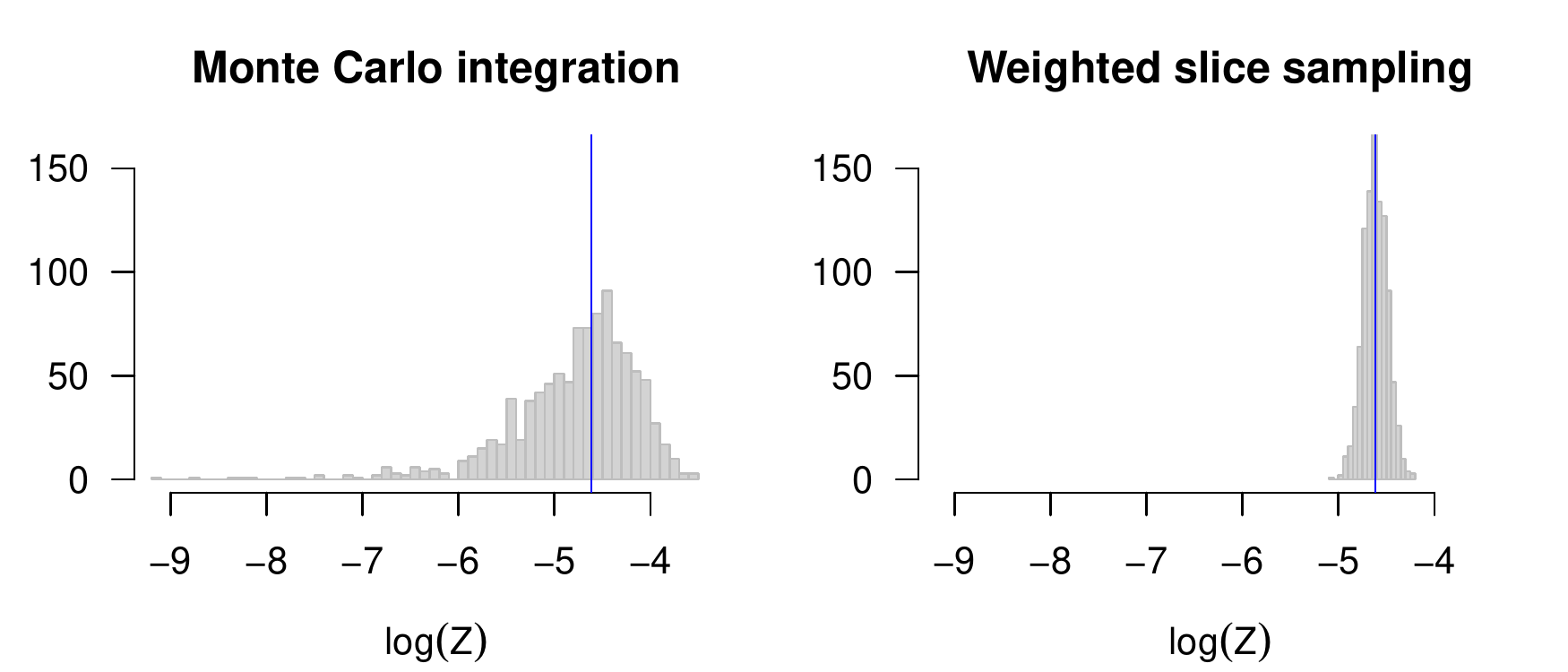}
\caption{\label{fig:normexp1} Monte Carlo variance for two estimators of $\log Z$ under 1000 repetitions.  Left: standard Monte Carlo integration with 10,000 draws from the prior.  Right: weighted slice sampling with the default weight function ($\eta = 10^{-4}$).  The vertical blue bar shows the true answer, $\log Z \approx -4.615$.  Each estimate used 10,000 Monte Carlo draws, with the first 500 discarded as burn-in under weighted slicing sampling.}
\end{figure}

As a simple toy problem, we consider the case of a sharp Gaussian likelihood integrated against a diffuse exponential prior.  Even in the one-dimensional case, however, this situation can pose difficulties for na\"ive Monte Carlo integration when the likelihood is sufficiently sharp.

Suppose that we observe data $a \sim N(x, \sigma^2)$.  To compute the slice region, we find the set of values $x$ such that $L(x) > y$, which happens whenever
$$
\frac{1}{\sqrt{2\pi \sigma^2}} \exp \left\{\frac{-(a - x)^2}{2\sigma^2} \right\} > y \, .
$$
This is equivalent to $x \in (l_y, a + \delta_y)$,
where
\begin{eqnarray*}
l_y &=& \max(0, a-\delta_y) \\
\delta_y &=& \sqrt{ -2 \sigma^2 \log(y \sqrt{2\pi\sigma^2})} \, ,
\end{eqnarray*}
leaving implicit the dependence of $\delta_y$ upon $\sigma$.

Now let $x \sim \mbox{Ex}(1/\tau)$; we imagine that the scale parameter $\tau$ is much larger than $\sigma$, such that draws from the prior are very unlikely to fall near a region of high likelihood.  To apply vertical-likelihood sampling, we first compute
$$
Z(y) = \int_{l_{y}}^{a+\delta_y} p(x) d x \, .
$$
For the exponential prior, this is available in closed form as
$$
Z(y) = \left\{
\begin{array}{ll}
1 - \exp \{-(a+\delta_y)/\tau \} \, , & a - \delta_\lambda \leq 0 \\
2 \exp(-a/\tau) \sinh(\delta_y/\tau) \, , & a - \delta_\lambda > 0 \, .
\end{array}
\right.
$$

This can also be inverted explicitly.  Let $c = 1 - \exp(-2a/\tau)$, and let $Z(y) = s$.  This may be solved for $y$ by first solving for $\delta_y$:
$$
\delta_y(s) = \left\{
\begin{array}{ll}
-\tau \log(1-s) \, , & s \geq c \\
\tau \sinh^{-1} \left\{ \frac{s}{2 \exp(-a/\tau)}  \right\}  \, , & s < c \, .
\end{array}
\right.
$$
From this, $y$ may easily be recovered via the relation $\delta_y^2  = -2 \sigma^2 \log(y \sqrt{2\pi\sigma^2})$.

Figure \ref{fig:normexp2} shows the implied importance functions corresponding to the weight function in (\ref{eqn:defaultWeightFunction}) for several different choices of $\eta$.  The prior is shown in blue, the likelihood in green, and the draws from $\pi_w(x)$ as a grey histogram.  Figure \ref{fig:normexp1} then shows the Monte Carlo variance of the estimator $\log \hat{Z}$, versus that of standard Monte Carlo integration with draws from the prior distribution.  Each estimate used 10,000 Monte Carlo draws, and the calculation was repeated 1000 times.  It is clear that the reweighting has significantly improved the accuracy and stability of the estimator, compared with na\"ive Monte Carlo integration.

\subsection{Multivariate $t$ likelihood, normal prior}

Our second example is to calculate the integral
$$
Z = \int_{\mathbb{R}^d}
\left ( 1 + \frac{ x^T x}{\nu} \right )^{ - \frac{1}{2} (\nu+d)}  \cdot  \left( \frac{\tau}{2 \pi} \right)^{d/2} \exp\{ - \tau  x^T x /2 \}  \ dx
$$
This is proportional to the normalizing constant of the posterior distribution arising from a multivariate Student-$t$ likelihood (ignoring the leading constants) and a mean-zero normal prior with prior precision $\tau I$.  We calculate this integral in $d=50$ dimensions, which is a challenge for most Monte Carlo methods.  This problem is a useful test case, because the exact value of $Z$ can be calculated using Kummer's confluent hypergeometric function of the second kind \citep[e.g.][Equation 13.1.3]{abramowitz1970}.  For details of calculating $Z$ using the Kummer function, see Appendix \ref{app:multT_example}.

We choose $\nu = 2$ and $\tau=1$, for which the correct answer is $Z \approx 1.95 \times 10^{-29}$.  For implementing weighted slice sampling, $Z(y)$ is easily derived in terms of the quantile and inverse quantile functions of the gamma distribution.

We used four Monte Carlo methods to approximate $Z$:
\begin{compactenum}
\item the harmonic mean estimator with 10000 samples and a burn-in of 1000.
\item Chib's estimator \citep{chib:1995} centered at $x^\star=0$, also known as Besag's candidate method, with 10000 samples and a burn-in period of 1000.
\item Nested sampling with $K=50$ and $n=10000$.
\item Vertical-likelihood Monte Carlo with the default weight function derived from the score function of $Z(y)$ (with $\eta = 0.01$), implemented via weighted slice sampling with 10000 samples and a burn-in of 1000.
\end{compactenum}
We repeated each method 100 times, to assess the root mean-squared Monte Carlo error (RMSE).

\begin{table}
\begin{center}
\caption{\label{tab:multivariateTresults}
Average estimate and root mean-squared error in 100 repetitions of each method for the multivariate $T$ example.  The true answer is $Z \approx 1.95 \times 10^{-29}$.  Weighted slice sampling has the smallest mean-squared error over 100 Monte Carlo repetitions, improving by a relative factor of $\approx 87\%$ over nested sampling.
}
\vspace{0.5\baselineskip}
\begin{tabular}{r r r r r r}
	&& Harmonic & Chib & Nested & Weighted slice \\
	\midrule
Average && $1.26 \times 10^{-13}$ & $1.35 \times 10^{-27}$ & $ 2.52 \times 10^{-29}$ & $1.61 \times 10^{-29}$ \\
RMSE && $6.87 \times 10^{-13}$       & $1.50 \times 10^{-27}$ & $ 1.87 \times 10^{-29}$ & $9.98 \times 10^{-30}$
\end{tabular}
\end{center}
\end{table}

The results are summarized in Table \ref{tab:multivariateTresults}.  The harmonic mean estimator is not competitive, for reasons that are widely known and discussed at length in \citet{wolpert:schmidler:2012}.  In this case, it is off by 16 orders of magnitude on average.  The Chib--Besag estimator is much better, but is still off by 2 orders of magnitude on average.  Nested sampling does much better, and indeed very well in objective terms.  Its average answer comes within $\approx 25\%$ of the true value, and its root mean-squared error is of the same order of magnitude as the average answer, suggesting that it rarely missed by much.  Vertical likelihood Monte Carlo, as implemented via weighted slice sampling, performs the best on this example, with an RMSE that is about $45\%$ smaller than that of nested sampling, in relative terms.

\section{Mixing properties}

\label{sec:mixing}

We have seen that both weighted slice sampling and nested sampling are two methods for operationalizing the score-function heuristic described earlier.  One major difference of weighted slice sampling is that, at any given iteration, the slice sampler can move up or down in likelihood-ordinate space.  Nested sampling only ever moves up.  As a result, slice sampling can be expected to mix better over likelihood ordinates.  An interesting comparison is with the diffusive nested sampling algorithm of \citet{brewer:etal:2011}, which also allows one to move down in likelihood space, and which seems to perform favorably versus ordinary nested sampling.

In fact, the mixing properties of weighted slice sampling can be characterized quite precisely.  The sampling scheme outlined in Section \ref{sec:weightfunction} produces samples $ ( X_{n} , U_{n} )$, and gives rise to a marginal chain $ \{ L(X_n)\}$.  Following results from \citet{roberts:rosenthal:2002}, the convergence properties of this marginal chain are governed by the transition kernel for a pair of likelihood ordinates $(y,z) \in L(\mathcal{X})$.  This may be calculated via
$$
\mathbb{P} \left\{ L( X_{n+1} ) \geq z \mid L(X_n)=y \right\}
 = \int_{ U_{n+1} } \mathbb{P} \left\{ L( X_{n+1} ) \geq z \mid  U_{n+1} \right\} p\left\{ U_{n+1} \mid L(X_n)=y \right\} d U_{n+1} \, .
$$
Recall that
$$
Z(u) = \int_{ x: L(x)>u } P(x) d x \; \; {\rm and} \; \; W(u) = \int_0^u w(s)d s \; .
$$
The weighted slice sampler has the following joint, marginal and conditional distributions
$$
P_{WS} ( x , u ) = w(u) \mathbb{I} \left (0<u<L(x)\right )P(x) / Z_w 
$$
where the normalization constant can be computed two ways
$$ 
Z_w = \int_{\mathcal{X}} W(L(x))P(x) d x = \int_0^\infty w(u) Z(u) du \, .
$$ 

The following lemma characterizes the transition kernel.
\begin{lemma}
For a general weighted slice sampler with weight function $w(u)$ and corresponding cumulative weight function $W(u)$, the transition kernel of the likelihood-ordinate chain is given by
$$
\mathbb{P} \left ( L( X_{n+1} ) \leq z \mid L(X_n)=y \right ) = \int_0^{ y \wedge z} \left \{ 1 - \frac{Z(z)}{Z(u)} \right \} \frac{\omega(u)}{\Omega(y)} d u \, .
$$
\end{lemma}

This follows directly from $ \pi(u \mid L(x)=y) = w(u)/W(y)$ and 
\begin{align*}
\mathbb{P} \left\{ L( X_{n+1} ) \leq z \mid L(X_n)=y \right\} & = \int_0^{ y \wedge z} F(L(x) < z \mid  u) \ \pi( u \mid L(x)=y) d u \\
& =  \int_0^{ y \wedge z} \left \{ 1 - \frac{Z(z)}{Z(u)} \right \} \frac{w(u)}{W(y)} d u \, .
\end{align*}

The following corollary is derived straightforward from this lemma.
\begin{corollary}
Suppose that the choice of weight function $ W(u)=Z(u)^{-1}$ defines a proper joint distribution. Then the corresponding transition matrix
reduces to
$$
\mathbb{P} \left\{ L( X_{n+1} ) \leq z \mid L(X_n)=y \right\} = \frac{ W( y \wedge z)}{ W(y)} - 
\frac{Z(z) W( y \wedge z)^2}{2 W(y) } 
$$
For $ z < y$ we have
\begin{align*}
\mathbb{P} \left\{ L( X_{n+1} ) \leq z \mid L(X_n)=y \right\} &= \frac{1}{2} \frac{Z(y)}{Z(z)} \\
\mathbb{P} \left\{ L( X_{n+1} ) \leq y \mid L(X_n)=y \right\}  &= \frac{1}{2}  \quad \forall y \quad {\rm and} \quad \forall P(x) \, .
\end{align*}
\end{corollary}

Hence our score-function heuristic produces a weight function with the stabilizing property that, on the next step of the marginal chain, one is equally likely to move up or down in likelihood-ordinate space.  From the results of \citet{roberts:rosenthal:2002}, it follows that the likelihood-ordinate chain is geometrically ergodic.  Note that when $ p(x)/Z\{L(x)\}$ is not integrable, we instead use the table-mountain-hat function $ W(u)= \min\{ \epsilon^{-1}, Z(u)^{-1} \}$.

\section{Discussion}

\label{sec:discussion}

Our review has sought to provide an answer to the question: what should the distribution of $Y\equiv L(X)$ look like under draws from the proposal $X \sim g(x)$ in importance sampling? We have provided a specific recommendation in the form of the score-function heuristic, and that both nested sampling and weighted slice sampling are methods for putting this principle into practice.  We also demonstrate a number of previously unappreciated connections among a wide class of methods---for example, by showing that the harmonic-mean estimator, the power-posterior method, energy-level methods from statistical physics, and nested sampling can all be characterized in terms of different choices for a one-dimensional weight function $w(u)$.  Of course, many practical details of implementation of the vertical-likelihood principle remain to be studied.  The ability of weighted slice sampler to mix so rapidly over likelihood-ordinate space may account for its slightly better performance on the multivariate-$t$ example.

We conclude by highlighting two possible extensions of the approach that can inform future work.  These two extensions, along with the convergence analysis of Section \ref{sec:mixing}, highlight one of the advantages of our approach: it inherits all of the tricks and theoretical machinery that have grown up around MCMC-based methods.  

First, it is common to use multiple slice variables as part of an MCMC---for example, if the likelihood can be factorized as $L(x) = L_1(x) L_2(x)$.  In this case, one could use the representation
$$
Z = \int_\mathcal{X} \left\{ \int_0^{L_1(x)} du_1 \right\} \left\{ \int_0^{L_2(x)} du_2 \right\} \ dP(x) \, ,
$$
and reweight the slice variables $(u_1, u_2)$ either independently or jointly.  Although we do not pursue the point here, this could potentially allow the approach to be extended to a much richer class of models with complicated likelihood functions \citep[e.g.][]{higdon:1998}.

Second, many complicated likelihoods have representations as mixtures of simpler densities, such as the multivariate normal.  Because it is an MCMC-based method, our approach can easily accommodate these extra latent variables.  For example, suppose that $L(x)$ has the representation
$$
L(x) = \int_0^{\infty} \exp(-cx^2/2) \ \pi(c) \ dc \, .
$$
Our approach can be modified accordingly.  Let $L(x,c) = \exp(-cx^2/2)$, and write Z as
\begin{eqnarray*}
Z &=& \int_\mathcal{X} \int_0^\infty  L(x,c)  \ \pi(c)  \ p(x) \ dc \ dx \\
&=& \int_\mathcal{X} \int_0^\infty  \left\{ \int_0^{L(x,c)} du \right\} \ \pi(c)  \ p(x) \ dc \ dx \, .
\end{eqnarray*}
We can now re-weight the slice variable $u$ according to the inverse of the individual slice normalization constants $Z(u,c)$, and proceed as before.

\appendix

\section{Proof of Theorem \ref{theorem:ordinatedistn}}

Let $m$ denote the maximum value of the likelihood, possibly infinite.  Let $F(\lambda)$ denote the CDF of the likelihood ordinate under $\pi_w$.  Clearly
\begin{align*}
F(\lambda) 
&= 1 - \mathbb{P}_{\pi_w} \left ( L(x) > \lambda \right ) \\
& = 1 - Z_w^{-1} \int_{L(x)> \lambda} L(x) \pi(x) d x = 1 - Z_w^{-1} \int_{L(x)> \lambda} \left \{ \int_0^{L(x)} w (s) ds \right \}  \pi( d x ) \\
& = 1 - Z_w^{-1} \int_0^m \left \{ \int_{ L(x) > \max(\lambda,s) } \pi(d x)  \right \} w(s) ds\\
& = 1 - Z_w^{-1} \left ( W(\lambda)  Z_w(\lambda) + \int_\lambda^m  \left \{ \int_{ L(x) > s } \pi(d x)  \right \} w (s) ds \right )\\
& = 1 - Z_w^{-1} W(\lambda) Z_w(\lambda) - Z_w^{-1} \int_\lambda^m w(s ) Z_w(s) ds \, .
\end{align*}
The ranges $ \lambda \leq L(x) $ and $ s \leq L(x) $ imply $ \max ( \lambda ,s ) \leq L(x) $ and $ 0< s< m$.
On the range $0<s<\lambda$ we have  $ \int_{ L(x) > \max(\lambda,s) } \pi(d x)  = Z_w(\lambda) $ as $ \max(\lambda,s)=\lambda$. The integral
over this range is $ W(\lambda) Z_w(\lambda)$. 

Differentiation then gives the result: $ W^\prime (\lambda) = w_\lambda $, and so
$$
\frac{d}{d\lambda}\left \{ W(\lambda) Z(\lambda) - \int_\lambda^m w(s ) Z(s) ds \right \}  
 = W(\lambda) Z^\prime (\lambda)
$$
as required.

\section{Multivariate $t$ example: details}
\label{app:multT_example}

The integral to be calculated is
$$
Z = \int_{\mathbb{R}^d}
\left ( 1 + \frac{ x^T x}{\nu} \right )^{ - \frac{1}{2} (\nu+d)}  \cdot  \left( \frac{\tau}{2 \pi} \right)^{d/2} \exp\{ - \tau  x^T x /2 \}  \ dx
$$
This is proportional to the normalizing constant of the posterior distribution arising from a multivariate Student-$t$ likelihood and a mean-zero normal prior with prior precision $\tau I$.  

Let $a=(\nu+d)/2$.  To calculate this integral in terms of hypergeometric functions, we exploit the following facts.  First,
$$
\left ( 1 + \frac{ x^T x}{\nu} \right )^{ - a}  = 
\int_0^{\infty} \exp\left( -\lambda \frac{x^T x}{\nu} \right) p(\lambda) \ d \lambda  \, ,
$$
where
$$
p(\lambda) = \frac{ \lambda^{a - 1} e^{-\lambda} } {\Gamma(a)} \, .
$$

By Fubini's theorem, we may therefore write $Z$ as
$$
Z = \left( \frac{\tau}{2 \pi} \right)^{d/2}   \int_{0}^{\infty} \int_{\mathbb{R}^d}   \exp \left( -\lambda \frac{x^T x}{\nu} \right)
 \frac{ \lambda^{a - 1} e^{-\lambda} } {\Gamma(a)} \
 \exp\left( - \frac{\tau}{2} x^T x \right)
   \ dx \ d\lambda \, .
$$
Integrating first over $x$, we find that
\begin{eqnarray*}
Z &=&  \left( \frac{\tau}{2 \pi} \right)^{d/2} \cdot \frac{1 } {\Gamma(a)} \int_0^{\infty} \lambda^{a-1} e^{-\lambda}
 \int_{\mathbb{R}^d} \exp \left\{  -\frac{1}{2} \left( \frac{2\lambda}{\nu} + \tau \right) x^Tx \right\} \ dx \ d \lambda \\
 &=& \frac{1 } {\Gamma(a)} \int_0^{\infty} \left( \frac{2\lambda}{\nu \tau}  +1\right)^{-d/2} \lambda^{a-1} e^{-\lambda} d \lambda
\end{eqnarray*}

Make a change of variables to $t=2\lambda/(nu\tau)$, and let $b=\nu/2 + 1$, $s=\nu\tau/2$.  Then
\begin{eqnarray*}
Z &=& \frac{s^a } {\Gamma(a)} \int_0^{\infty} \left( t  +1\right)^{b-a-1} t^{a-1} e^{-s t} d t \\
&=& s^a U(a,b,s) \, ,
\end{eqnarray*}
where $U(a,b,s)$ is Kummer's confluent hypergeometric function of the second kind.  The choice $\nu = 2$, $\tau=1$, and $d=50$ ensures that we evaluate the Kummer function at integer arguments, for which accurate numerical routines exist.

\section{Adaptive Wang--Landau}

The goal of the adapative Wang-Landau algorithm is to try to learn the importance function: the target distribution $ \tilde{P}$ in the MCMC is always changing.  The adaptive Wang-Landau algorithm samples from
$$ \tilde{P}(x) = P(x) \times \frac{1}{d} \sum_{i=1}^d \frac{\mathbb{I}_{\mathcal{X}_i}(x)}{\int_{\mathcal{X}_i} P(x) dx},$$
where $d$ is the number of energy levels and $\mathcal{X}_i$ are a partition of the sample space $\mathcal{X}$. In other words, it tries to uniformly visit energy levels. For the goal, one needs to estimate the quantities
$$ \theta_i = \int_{\mathcal{X}_i} d P(x) \, .$$
The adaptive Wang-Landau algorithm initially sets $\theta_i^{(0)} \propto 1$ for all $i$. After sampling $X_t$ at $t$-th iteration, $\theta_i$ are updated via
$$ \log \theta_i^{(t)} = \log \theta_i^{(t-1)} + \gamma_t \{ \mathbb{I}_{\mathcal{X}_i} (X_t) - 1/d \}$$
for some $\gamma_t$.

\end{spacing}

\begin{small}
\singlespacing
\bibliographystyle{abbrvnat}
\bibliography{masterbib}
\end{small}

\end{document}